\documentclass[sigconf]{acmart}
\AtBeginDocument{%
  }

\usepackage{amsmath, bbm}
\usepackage{amsfonts}
\usepackage{amsthm}
\usepackage{algorithm}
\usepackage{algorithmic}
\usepackage{nicefrac}
\usepackage{bm}
\usepackage{hyperref}
\usepackage[utf8]{inputenc} 
\usepackage[T1]{fontenc}
\usepackage{graphicx}
\usepackage[font=small]{caption}
\usepackage{tikz}
\usetikzlibrary{positioning}
\usepackage{pgfplots}
\usepackage{epstopdf}
\usepackage[labelformat=simple]{subcaption}
\usepackage{multirow}
\usepackage{lipsum}
\usepackage{mathtools}

\long\def\comment#1{}

\newtheorem{theorem}{Theorem}

\newcommand \COMMENTS[1] {{\color{gray}\textit{/* #1 */}}} 

\input{mysymbol.sty}

\pgfplotsset{compat=1.18}

\setcopyright{acmlicensed}
\copyrightyear{2025}
\acmYear{2025}
\acmDOI{XXXXXXX.XXXXXXX}
\acmConference[Mobihoc '25]{The 26th International Symposium on Theory, Algorithmic Foundations, and Protocol Design for Mobile Networks and Mobile Computing}{October 27-30, 2025}{Houston, USA}
\acmISBN{979-8-4007-XXXX-X/2025/10}

\begin{document}

\title{Generalizing Biased Backpressure Routing and Scheduling to  Wireless Multi-hop Networks with Advanced Air-interfaces}

\author{Zhongyuan Zhao}
\orcid{0000-0003-0346-8015}
\affiliation{%
  \institution{Rice University}
  \city{Houston}
  \state{Texas}
  \country{USA}
}
\email{zhongyuan.zhao@rice.edu}

\author{Yujun Ming}
\affiliation{%
  \institution{Rice University}
  \city{Houston}
  \state{Texas}
  \country{USA}
}
\email{yujun.ming@rice.edu}

\author{Ananthram Swami}
\affiliation{%
  \institution{DEVCOM Army Research Laboratory}
  \city{Adelphi}
  \state{Maryland}
  \country{USA}
}
\email{ananthram.swami.civ@army.mil}

\author{Kevin Chan}
\affiliation{%
  \institution{DEVCOM Army Research Laboratory}
  \city{Adelphi}
  \state{Maryland}
  \country{USA}
}
\email{kevin.s.chan.civ@army.mil}

\author{Fikadu Dagefu}
\affiliation{%
  \institution{DEVCOM Army Research Laboratory}
  \city{Adelphi}
  \state{Maryland}
  \country{USA}
}
\email{fikadu.t.dagefu.civ@army.mil}

\author{Santiago Segarra}
\affiliation{%
  \institution{Rice University}
  \city{Houston}
  \state{Texas}
  \country{USA}
}
\email{segarra@rice.edu}

\renewcommand{\shortauthors}{Zhao et al.}

\begin{abstract}
Backpressure (BP) routing and scheduling is a well-established resource allocation method for wireless multi-hop networks, known for its fully distributed operations and proven maximum queue stability. 
Recent advances in shortest path-biased BP routing (SP-BP) mitigate shortcomings such as slow startup and random walk, but exclusive link-level commodity selection still suffers from the last-packet problem and bandwidth underutilization. 
Moreover, classic BP routing implicitly assumes single-input-single-output (SISO) transceivers, which can lead to the same packets being scheduled on multiple outgoing links for multiple-input-multiple-output (MIMO) transceivers, causing detouring and looping in MIMO networks.
In this paper, we revisit the foundational Lyapunov drift theory underlying BP routing and demonstrate that exclusive commodity selection is unnecessary, and instead propose a Max-Utility link-sharing method. 
Additionally, we generalize MaxWeight scheduling to MIMO networks by introducing attributed capacity hypergraphs (ACH), an extension of traditional conflict graphs for SISO networks, and by incorporating backlog reassignment into scheduling iterations to prevent redundant packet routing. 
Numerical evaluations show that our approach substantially mitigates the last-packet problem in state-of-the-art (SOTA) SP-BP under lightweight traffic, and slightly expands the network capacity region for heavier traffic.
\end{abstract}

\ccsdesc[500]{Networks~Network resources allocation}
\ccsdesc[500]{Networks~Packet scheduling}
\ccsdesc[300]{Theory of computation~Self-organization}
\ccsdesc[300]{Theory of computation~Shortest paths}
\ccsdesc[500]{Networks~Cross-layer protocols}
\ccsdesc[500]{Theory of computation~Scheduling algorithms}
\ccsdesc[500]{Theory of computation~Stochastic approximation}
\ccsdesc[500]{Theory of computation~Routing and network design problems}

\keywords{Backpressure routing, queueing networks, Lyapunov drift, MIMO, MaxWeight scheduling, distributed algorithms}

\maketitle

\section{Introduction}\label{sec:intro}

Backpressure (BP) routing \cite{tassiulas1992, neely2005dynamic, georgiadis2006resource} is a well-established resource allocation method for wireless multi-hop networks, applicable to mobile ad-hoc/sensor networks, xG (device-to-device, wireless backhaul, and non-terrestrial coverage), vehicle-to-everything (V2X), Internet-of-Things (IoT), and machine-to-machine (M2M) communications \cite{kott2016internet,akyildiz20206g,chen2021massive,noor20226g}.
BP routing relies on a per-destination queuing system, where packets destined to node $c$ are referred as commodity $c$.
Characterized by exclusively selecting one commodity for each link, and MaxWeight  scheduling \cite{tassiulas1992}, 
BP routing allows packets to explore all possible routes towards their destinations, addressing inter-flow interference via joint routing and scheduling. 
Theoretically, it has proven maximum queue stability within the network capacity region, named as throughput optimality~\cite{neely2005dynamic, georgiadis2006resource, neely2022stochastic}. 
Practically, BP routing supports fully distributed implementation with the help of distributed schedulers \cite{joo2012local,zhao2021icassp,zhao2022twc,zhao2023graphbased}, providing efficiency, scalability, and robustness against single-point-of-failure, which are critical for military and disaster relief communications, and backbone infrastructures for social and economic activities.
In addition, the BP algorithm also applies to other queueing networks like edge/cloud computing \cite{Kamran2022deco,lin2020distributed,zhao2025icassp} and traffic signal control \cite{levin2023maxpressure}.

The original BP scheme has some well-known shortcomings like slow startup, random walk, and the last-packet problem~\cite{neely2005dynamic,georgiadis2006resource,neely2022stochastic,jiao2015virtual,cui2016enhancing,gao2018bias}. 
Therefore, various improvements have been proposed, such as queue-agnostic biases~\cite{neely2005dynamic, georgiadis2006resource, neely2022stochastic, jiao2015virtual, zhao2023icassp, zhao2023enhanced, zhao2024tmlcn}, virtual queues~\cite{ji2012delay, cui2016enhancing, hai2018delay, zhao2024tmlcn}, and route restrictions~\cite{ying2010combining, Rai2017loop, yin2017improving}. 
Shortest path-biased BP routing (SP-BP)~\cite{neely2005dynamic, georgiadis2006resource} inherits the throughput optimality of the original BP, while SOTA SP-BP~\cite{zhao2023icassp, zhao2023enhanced, zhao2024tmlcn} resolves the slow startup and random walk problems, improving latency and throughput with minimal additional overheads. 
In~\cite{zhao2023icassp, zhao2023enhanced, zhao2024tmlcn}, the shortest path bias is found via an all-pairs shortest path algorithm on a connectivity graph, with edge weights configured based on link features like scheduling duty cycle and long-term link rate. 
However, the last-packet problem still persists due to exclusive commodity selection, which can starve short-lived traffic that lacks consistent pressure from new packets~\cite{Alresaini2016bp,ji2012delay,Erfaniantaghvayi2024}.
The existing solutions often employ virtual queues to prioritize older packets~\cite{Alresaini2016bp,ji2012delay,zhao2024tmlcn} or separated queueing systems~\cite{Erfaniantaghvayi2024}, which can shrink the network capacity region~\cite{Alresaini2016bp,ji2012delay,zhao2024tmlcn,Erfaniantaghvayi2024}.

These problems become more consequential with increasingly diverse data traffic and advanced air-interfaces with high bandwidth and multiple-input-multiple-output (MIMO) capabilities.
The exclusive commodity selection can cause underutilization of high-bandwidth air-interfaces, especially in M2M~\cite{chen2021massive} and IoT~\cite{kott2016internet} applications with numerous bursty traffic.
This problem cannot be fully addressed by employing larger packet sizes or shorter time slots, due to inherent limits in channel coding, synchronization, and transmit/receive (T/R) switch in radio frontends.

Furthermore, we observe that classic BP schemes rely on an implicit assumption of single-input-single-output (SISO) networks.
Prior to link scheduling, a packet can be assigned to multiple outgoing links, leaving its path to be determined later by MaxWeight scheduling based on link utility and conflict relationship.
However, with MIMO transceivers, multiple outgoing links can be scheduled simultaneously.
In classic BP, the separated initial assignment and MaxWeight scheduling allows packets to travel over suboptimal routes, causing detouring and looping that render MIMO counterproductive.
These limitations hinder the popularity of BP routing.

In this paper, we seek to improve the basic operations in nearly all BP schemes by revisiting the  Lyapunov drift theory~\cite{neely2005dynamic,georgiadis2006resource,neely2022stochastic} and conflict modeling of MaxWeight scheduling~\cite{joo2012local,zhao2021icassp,zhao2022twc}.
Through careful analysis of the constraints in minimizing Lyapunov drift, we show that  exclusive commodity selection is unnecessary, and propose a link-sharing approach for  maximum utility (MaxU) multi-commodity selection. 
To extend MaxWeight scheduling into MIMO networks, we generalize traditional conflict graphs into attributed capacity hypergraphs (ACH), incorporating rate reassignment into each scheduling iteration to prevent redundant routing, and extend the distributed local greedy scheduler (LGS)~\cite{joo2012local} to ACH with a transceiver-level algorithm, LGS-MIMO, readily implementable as a network protocol.
Integrating these innovations into SOTA SP-BP~\cite{zhao2024tmlcn}, our MaxU SP-BP can significantly mitigate the last-packet problem while slightly expanding the network capacity region.

\noindent
{\bf Contribution:} The contributions of this paper are as follows:\\
1) By analyzing the constraints in minimizing Lyapunov drift, we show that exclusive commodity selection in classic BP is unnecessary and propose a link-sharing approach, MaxU, to address the long-standing last-packet problem and bandwidth underutilization under diverse traffic, keeping throughput optimality intact.\\
2) We extend BP routing into MIMO networks by modeling conflicts in a MIMO network as an ACH, incorporating rate reassignment into scheduling iteration to avoid redundant packet routing, and developing a transceiver-level algorithm for distributed MaxWeight scheduling in MIMO networks.\\
3) We show that, theoretically and empirically, MaxU SP-BP outperforms classic SP-BP in the network capacity region. \\
4) Through simulation, we demonstrate that our MaxU commodity selection and LGS-MIMO can effectively mitigate the last-packet problem and improve resource utilization for SP-BP under mixed streaming and bursty traffic in both SISO and MIMO networks.

\noindent
{\bf Notation:} 
$ |\cdot| $ represents the cardinality of a set.
$ \mathbbm{1}(\cdot) $ is the indicator function.
$ \mathbb{E}(\cdot) $ stands for expectation.
Upright bold lower-case symbol, e.g., $\bbz$, denotes a column vector, and $\bbz_i$ denotes the $i$-th element of vector $\bbz$. 
Upright bold upper-case symbol $\bbZ$ denotes a matrix, and $\bbZ_{i,j}$ its element at row $i$ and column $j$.
$\delta_{\ccalG}(i)$ stands for the set of edges incident on node $i$ on graph $\ccalG$, and sets $ \delta^{-}_{\ccalG}(i),\delta^{+}_{\ccalG}(i) $ stand for the outgoing and incoming edges of $i$ on directed $\ccalG$.
$ \ccalN_{\ccalG}(i) $ and $ \ccalN^1_{\ccalG}(i) $ denote the set of immediate neighbors of node $i$ or a node set (connected by pair-wise edges) on graph (hypergraph) $\ccalG$.

\section{Problem Formulation}
\label{sec:sys}

We model a wireless multi-hop network as a \emph{connectivity graph} $\ccalG^{n}$ and a \emph{conflict graph}~$\ccalG^{c}$. 
The connectivity graph $\ccalG^{n}=(\ccalV, \ccalE)$ is a directed graph, in which a node $i\in\ccalV$ represents a wireless device and $\ccalE$ is a set of directed links, in which $e=(i,j)\in\ccalE$ for $i,j\in\ccalV$ indicates that node $i$ can transmit data to node $j$ directly.
We assume that $\ccalG^{n}$ is strongly connected, i.e., a directed path always exists between any pair of nodes.
Although we consider the more general directed case, we typically have that $(j,i)\in\ccalE$ if $(i,j)\in\ccalE$.

The conflict graph, $\ccalG^c=(\ccalE,\ccalH)$, 
describes the conflict relationship between links (here, we re-use the notation $\ccalE$ referring to its elements as vertices) and is defined as follows: each vertex $e\in\ccalE$ corresponds to a link in $\ccalG^{n}$ and each undirected edge $(e_1, e_2)\in\ccalH$ indicates a conflict between links $e_1, e_2\in\ccalE$ in $\ccalG^{n}$.
Two links could be in conflict due to either 1) \emph{interface conflict}, i.e., two links share a wireless device with only one transceiver; or 2) \emph{wireless interference}, i.e., their incident devices are within a certain distance such that their simultaneous transmission will cause the outage probability to exceed a prescribed level.
Such a conflict graph is for SISO networks and assumed to be known, e.g., via channel monitoring~\cite[Supp]{zhao2022twc}, or more sophisticated approaches~\cite{yang2016learning}. 

We consider a time-slotted medium access control (MAC) in the wireless network.
Vector $\grave{\bbr}(t)\in\reals^{|\ccalE|}$ collects the real-time link rates 
of all wireless links at time step $t$, $1\leq t\leq T$, where $\grave{\bbr}_{ij}(t)$ is the instantaneous link rate of $(i,j)\in\ccalE$.
Vector $\bbr=\mathbbm{E}_{t}[\grave{\bbr}(t)]\in\reals^{|\ccalE|}$ collects the long-term average link rates of all wireless links. 
The unit of link rate is the number of packets per time slot.

Data packets with the same destination $c$ are referred as commodity $c$.
We consider two commodity types: 1) regular communications, e.g., each flow has a prescribed physical destination $c\in\ccalV$, and 2) computation offloading, where the destination of a data packet is a virtual sink $c\in\ccalW$ that only accepts packets of commodity $c$, and the computation is modeled as sending a packet to a virtual sink over a virtual link~\cite{zhao2025icassp}.
Notice that $\ccalV\cap\ccalW=\emptyset$. 
We denote the set of all commodities as $\ccalC=\ccalV\cup \ccalW$.

\textbf{Queueing system:} Each device hosts per-commodity queues in a first-in-first-out (FIFO) fashion.
The length of the queue for commodity $c$ on device $i\in\ccalV$ at the beginning of time slot $t$ is denoted by $Q_i^{(c)}(t)$. 
The evolution of queue length of a commodity $c\in\ccalC$ on a device $i\in\ccalV$ can be expressed as:
\begin{subequations}\label{E:queue}
\begin{align}
Q_{i}^{(c)}(t+1) &= Q_{i}^{(c)}(t) - M_{i-}^{(c)}(t) + M_{i+}^{(c)}(t) + A_{i}^{(c)}(t),\\
M_{i-}^{(c)}(t) &= \sum_{j\in\ccalV} \mu_{ij}^{(c)}(t),\;\; M_{i+}^{(c)}(t) = \sum_{j\in\ccalV} \mu_{ji}^{(c)}(t) \;,
\end{align}
\end{subequations}
where $ M_{i-}^{(c)}(t), M_{i+}^{(c)}(t) $ are the total outgoing and incoming packets of commodity $c$ on device $i$ in time slot $t$, 
$\mu_{ij}^{(c)}(t)$ is the number of packets of commodity $c$ to be transmitted over link $(i,j)$ in time slot $t$, and $\mu_{ij}^{(c)}(t)=0$ if $(i,j)\notin \ccalE$, 
and $A_{i}^{(c)}(t)$ stands for the number of new packets of commodity $c$ generated by users on device $i$ in time slot $t$. 
In SP-BP~\cite{neely2005dynamic, georgiadis2006resource, neely2022stochastic, zhao2023icassp, zhao2023enhanced, zhao2024tmlcn}, biased backlog is defined as
\begin{equation}\label{E:formulation:bb}   
 U_{i}^{(c)}(t) = Q_{i}^{(c)}(t) + B_{i}^{(c)}, \;\forall\; i\in\ccalV,c\in\ccalC\;, 
\end{equation}
where $ 0\!\leq \! B_{i}^{(c)}\!\!<\!\infty $ is a queue-agnostic, non-negative bias, representing the shortest path distance from $i$ to $c$~\cite{neely2005dynamic,zhao2023icassp,zhao2023enhanced,zhao2024tmlcn}.
For example, $\{B_i^{(c)}\}_{i\in\ccalV, c\in\ccalC}$ can be found by SOTA weighted all-pairs shortest path algorithm in near linear time~\cite{bernstein2019distributed}.
The backpressure of commodity $c$ on link $(i,j)$ is defined as 
\begin{equation}\label{E:backpressure}
U_{ij}^{(c)}(t)= U_{i}^{(c)}(t) - U_{j}^{(c)}(t), \;\forall\; (i,j)\in\ccalE, c\in\ccalC\;.    
\end{equation} 

\subsection{Lyapunov Drift Minimization}

The operations of SP-BP are designed to find the rate assignments for all link-commodity pairs (routing and scheduling decisions) at every time step $t$, in order to maintain queue stability within the network capacity region, namely, throughput optimality~\cite{zhao2024tmlcn,neely2005dynamic,georgiadis2006resource,neely2022stochastic}.
This can be achieved by minimizing the upper bound of the Lyapunov drift of the biased backlogs -- equivalent to solving the following optimization for every time slot $t$~\cite[Theorem 1]{zhao2024tmlcn}:
\begin{subequations}\label{E:formulation}
	\begin{align}
		\bbM^{*}(t) &= \argmax_{\bbM(t)} \sum_{i\in\ccalV}\sum_{j\in\ccalV}\sum_{c\in\ccalC} \mu_{ij}^{(c)}\!(t) U_{ij}^{(c)}\!(t)\label{E:formulation:obj}\\
		\text{s.t. } 
		  \bbM(t) & = \left[\mu_{ij}^{(c)}(t) \mid i,j\in\ccalV, c\in\ccalC \right] \;, \label{E:formulation:M}  \\
		 \mu_{ij}^{(c)}\!(t)&\geq 0,\;\forall\; (i,j)\in\ccalE, c\in\ccalC;\;  \mu_{ij}^{(c)}\!(t)=0,\forall\;(i,j)\notin\ccalE, \label{E:formulation:mu} \\
		 \grave{\bbr}_{ij}(t) &\geq  \sum_{c\in\ccalC} \mu_{ij}^{(c)}\!(t), \;\forall\; (i,j)\in\ccalE  \;, \label{E:formulation:link} \\ 
		 Q_{i}^{(c)}(t) & \geq \sum_{j\in\ccalV} \mu_{ij}^{(c)}(t), \;\forall\; i\in\ccalV, c\in\ccalC  \;, \label{E:formulation:comm} \\
      1 & \geq \bbx_{e_1}(t)+\bbx_{e_2}(t),\; \forall\; (e_1,e_2)\in\ccalH\;, \label{E:formulation:conflict}\\
      \bbx_{e}(t)&=\mathbbm{1}\bigg( \sum_{c\in\ccalC}\mu_{ij}^{(c)}(t) \bigg) \;, e=(i,j)\in\ccalE\;. \label{E:formulation:x}
	\end{align} 
\end{subequations}
In \eqref{E:formulation:obj}, the objective function seeks to maximize the rate for those links and commodities with higher backpressure [cf.~\eqref{E:backpressure}], hence the name backpressure routing.
The constraints in \eqref{E:formulation} are explained as follows:
\eqref{E:formulation:M} specifies the decision variables as the \emph{final} rate assignments for all commodity-link pairs, e.g.,~\eqref{E:quota} or~\eqref{E:sch-mimo}.
\eqref{E:formulation:mu} states that the rate assignment of a commodity-link pair should be non-negative, or zero if links do not exist. 
\eqref{E:formulation:link} states that the sum rate of all commodities on a link should not exceed the real-time link rate. 
\eqref{E:formulation:comm} states that the sum rate of a commodity across the outgoing links of a device should not exceed its queue length.
\eqref{E:formulation:conflict} and \eqref{E:formulation:x} prohibit simultaneous transmissions on conflicting links.

\subsection{Operations of Classic SP-BP}\label{sec:sys:sp-bp}
It is difficult to exactly solve the optimization in~\eqref{E:formulation}, therefore existing solutions rely on heuristics. 
The classic SP-BP is found by replacing the constraints in~\eqref{E:formulation:link} and \eqref{E:formulation:comm} with a stricter one, i.e., only one commodity per link is allowed in each time slot:
\begin{equation}\label{E:optcmdy}
    \sum_{c\in\ccalC}\mathbbm{1}\Big(\mu_{ij}^{(c)}(t)>0\Big) \leq 1,\;\forall\; (i,j)\in \ccalE\;.
\end{equation}
This leads to the classic 4-step BP routing and scheduling~\cite{neely2005dynamic, georgiadis2006resource, neely2022stochastic, zhao2023icassp, zhao2023enhanced, zhao2024tmlcn}, which are expressed in a slightly different form for the convenience of proof and comparison.

\noindent \textbf{Step 1}. The optimal commodity $c_{ij}^{*}(t)$ on each {directed} link (${i,j}$) is selected as the one with the maximal backpressure, i.e., 
\begin{equation}\label{E:commodity}
    c_{ij}^{*}(t) =\argmax_{c\in\ccalV}\; U_{ij}^{(c)}(t) \;,
\end{equation}

\noindent \textbf{Step 2}. All of the real-time link rate $\grave{\bbr}_{ij}(t)$ of link $(i,j)$ is initially assigned to its optimal commodity $c_{ij}^{*}(t)$,
\begin{equation}\label{E:gamma:original}
    \gamma_{ij}^{(c)}\!(t) = \begin{cases}
         \min\!\left\{\grave{\bbr}_{ij}(t),Q_{i}^{(c)}\!(t)\right\}, & \text{if } c=c_{ij}^{*}\!(t), U^{(c)}_{ij}\!(t)>0, \\
         0, & \text{otherwise}.
    \end{cases}    
\end{equation}
Then, compute the link utility vector $\bbw(t)\!=\!\big[w_{ij}(t) \mid (i,j)\in\ccalE\big]$, where the utility of link (${i,j}$) is found as
\begin{equation}\label{E:weight}
    w_{ij}(t) =\sum_{c\in\ccalC} \gamma_{ij}^{(c)}(t)\cdot\max\left\{U_{ij}^{(c)}(t),0 \right\}.
\end{equation}

\noindent \textbf{Step 3}. MaxWeight scheduling \cite{tassiulas1992} finds the schedule $\bbx(t)\in\{0,1\}^{|\ccalE|}$ in order to activate a set of \emph{non-conflicting links} achieving the maximum total utility as follows: 
\begin{subequations}\label{E:scheduling}
\begin{align}
    \bbx (t) &= \argmax_{\tilde{\bbx} (t)\in \{0,1\}^{|\ccalE|} } ~ \tilde{\bbx}(t)^\top  \bbw(t) \;,\label{E:scheduling:obj} \\
    \text{s.t., } \tilde{\bbx}_{e_1}(t) &+ \tilde{\bbx}_{e_2}(t)\leq 1,\;\forall\; (e_1,e_2)\in\ccalH\;,\label{E:scheduling:conflict}
\end{align}
\end{subequations}
where the constraint in~\eqref{E:scheduling:conflict} states that conflicting links should not be scheduled together.
MaxWeight scheduling involves solving an NP-hard maximum weighted independent set (MWIS) problem \cite{joo2010complexity} on the conflict graph. 
In practice, \eqref{E:scheduling} is solved approximately by distributed heuristics, such as the LGS~\cite{joo2012local} and its GCN-based enhancements~\cite{zhao2021icassp, zhao2022twc, zhao2023graphbased}.
In our tests, LGS is used.

\noindent \textbf{Step 4}. 
Apply the schedule to the initial assignment to obtain the final assignment for both routing and link scheduling
\begin{equation}\label{E:quota}
    \mu_{ij}^{(c)}\!(t) = {\gamma_{ij}^{(c)}\!(t)}\cdot \bbx_{ij} (t) ,\;\forall (i,j)\in\ccalE, c\in\ccalC.
\end{equation}

\section{Link-Sharing Commodity Selection}\label{sec:solution}

The legacy exclusive commodity selection~\cite{tassiulas1992,neely2005dynamic,georgiadis2006resource} described in ~\eqref{E:optcmdy} can be traced back to an era when traffic was mostly streaming and physical layer bandwidth was limited such that each transmission carried only one packet.
However, as the bandwidth of air-interfaces has grown by orders of magnitude and while packet size and duration of time slot cannot scale accordingly due to limits in channel coding, synchronization, and T/R switch in radio frontends, more data packets must be enclosed in each transmission, rendering the exclusive commodity selection inefficient.
Furthermore, larger network sizes and increasingly diverse traffic types that lead to numerous commodities and a mixture of short-lived and streaming traffic, also exacerbate the consequences of bandwidth wastage and the last-packet problem~\cite{Alresaini2016bp,ji2012delay} in classic BP schemes.

By laying out the actual constraints of routing and scheduling decisions for throughput optimality in~\eqref{E:formulation}, it is evident that the legacy constraint in~\eqref{E:optcmdy} is overly restrictive.
Based on the looser constraints in~\eqref{E:formulation:link} and \eqref{E:formulation:comm}, we develop link-sharing commodity selection to mitigate the drawbacks of the legacy exclusive approach.

\subsection{Maximum Utility (MaxU) Rate Assignment}
\label{sec:maxu}
To maximize the link utility while strictly enforcing \eqref{E:formulation:comm}, we can replace the 4 steps of SP-BP in~\eqref{E:commodity}--\eqref{E:quota} as follows:

\noindent\textbf{Step 1.}
Filter out commodities $c\in\ccalC$ on each link $(i,j)\in\ccalE$:
\begin{equation}\label{E:mu:select}
\alpha_{ij}^{(c)}(t)= \mathbbm{1}\!\left( U_{ij}^{(c)}\!(t) > 0 \right)\cdot\mathbbm{1}\!\left(Q_i^{(c)}\!(t)>0\right)\;.
\end{equation}

\noindent\textbf{Step 2.}
On each directed link $(i,j)$, first sort the remaining commodities decreasingly by backpressure: 
\begin{equation}\label{E:sort}
    \tilde{\bbc} = \underset{c\in\ccalC}{\arg\textup{sort}_{\downarrow}} \left[\alpha_{ij}^{(c)}(t) \cdot U_{ij}^{(c)}(t)\right] \;,  
\end{equation}
where vector $\tilde{\bbc}$ collects all commodities in sorted order, and $\tilde{\bbc}_n$ is its $n$th element.
Next, sequentially allocate the residual link rate ${\acute{\bbr}}_{ij}(t)$ to the remaining commodities in sorted order $n=1,2,\dots,|\ccalC|$:
\begin{subequations}\label{E:allocate}
\begin{align}
    \acute{\bbr}_{ij}(t) &= \max\left\{\grave{\bbr}_{ij}(t) -\!\! \sum_{m=1}^{n-1}\gamma_{ij}^{(\tilde{\bbc}_m)}\!(t), 0\right\} \;,\\
    \gamma_{ij}^{(\tilde{\bbc}_n)}\!(t) &= \min\left\{\alpha_{ij}^{(\tilde{\bbc}_n)}\!(t)\cdot\acute{\bbr}_{ij}(t),\; Q_i^{(\tilde{\bbc}_n)}\!(t)\right\}\;.
\end{align}
\end{subequations}
Then, compute link utilities $\bbw(t)$ with \eqref{E:weight} to finish step 2.

\noindent\textbf{Steps 3-4.}
To find the final schedule $\bbx(t)$ and rate assignment $\bbM(t)$, we could follow the steps in~\eqref{E:scheduling} and~\eqref{E:quota} for SISO networks, 
or the following general link scheduling function detailed in Section~\ref{sec:schedule}: 
\begin{equation}\label{E:sch-mimo}
[\bbM(t),\bbx(t)]=\psi\left(\bbw(t),\grave{\bbr}(t), \bbGamma(t),\bbQ(t),\bbU(t);\ccalG^n,\ccalG^c, K\right)\;, 
\end{equation}
where $\ccalG^n$ and $\ccalG^c$ are respectively the connectivity graph and ACH, $K$ sets the maximum number of iterations to limit scheduling overhead, and matrices $\bbGamma(t),\bbU(t)\in\reals^{|\ccalE|\times|\ccalC|}$ respectively collect $ \gamma_{e}^{(c)}(t), U^{(c)}_{e}(t)$ for all $ e\in\ccalE,c\in\ccalC $, and matrix $\bbQ(t)\in\mathbb{Z}_{\geq 0}^{|\ccalV|\times|\ccalC|}$ collects queue length $ Q^{(c)}_{i}(t) $ for all $ i\in\ccalV,c\in\ccalC $.

\begin{figure}[t]
    \centering
    \includegraphics[width=0.9\linewidth]{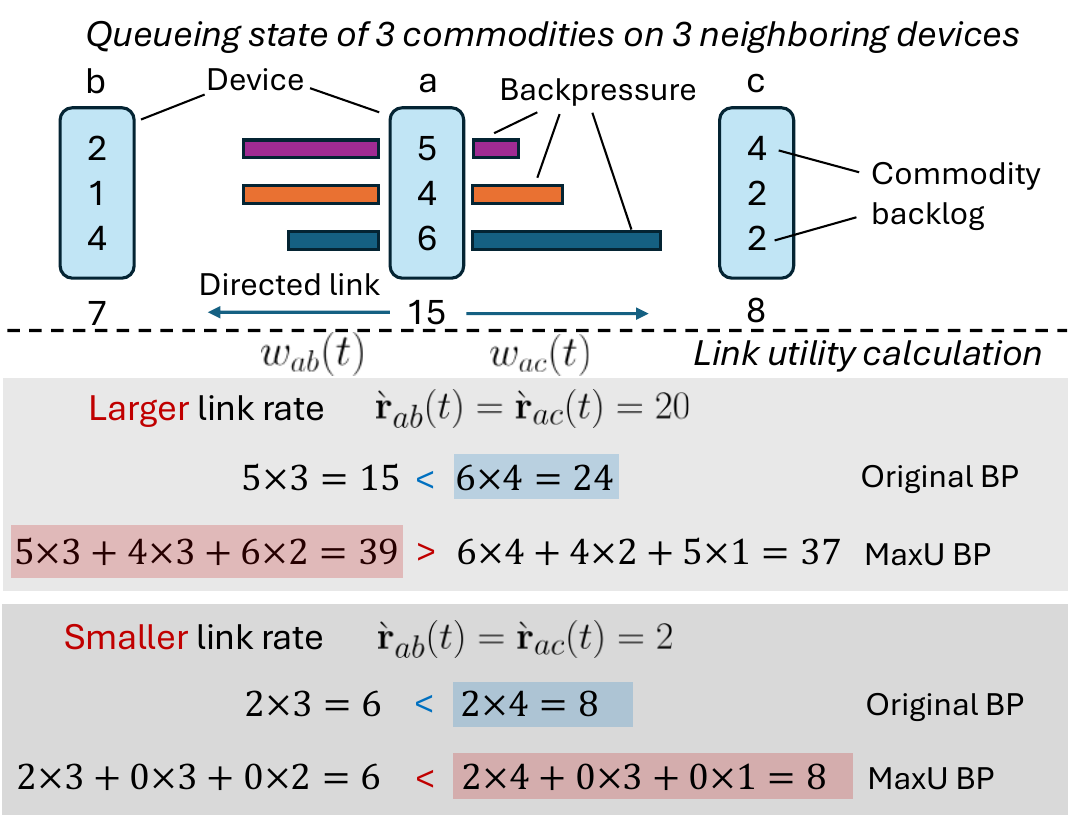}
    \vspace{-0.1in}
    \caption{A mini example of commodity selection and link utility calculation with three devices and three commodities.}
    \label{fig:linksharing}
    \Description{A mini example of commodity selection approaches}
\end{figure}

\begin{figure*}[!t]
	\centering
    \hspace{-3mm}
    \subfloat[]{
    \includegraphics[height=1.65in]{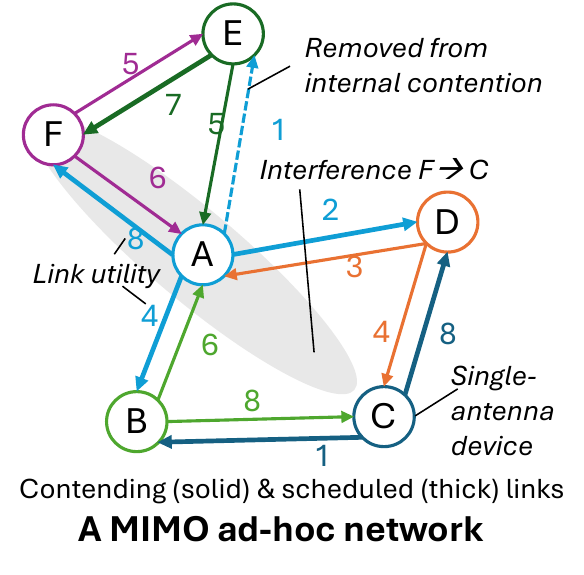}
    \label{fig:graph:network}
	  \vspace{-0.1in}
    }\hspace{-2mm}
    \subfloat[]{
    \includegraphics[height=1.65in]{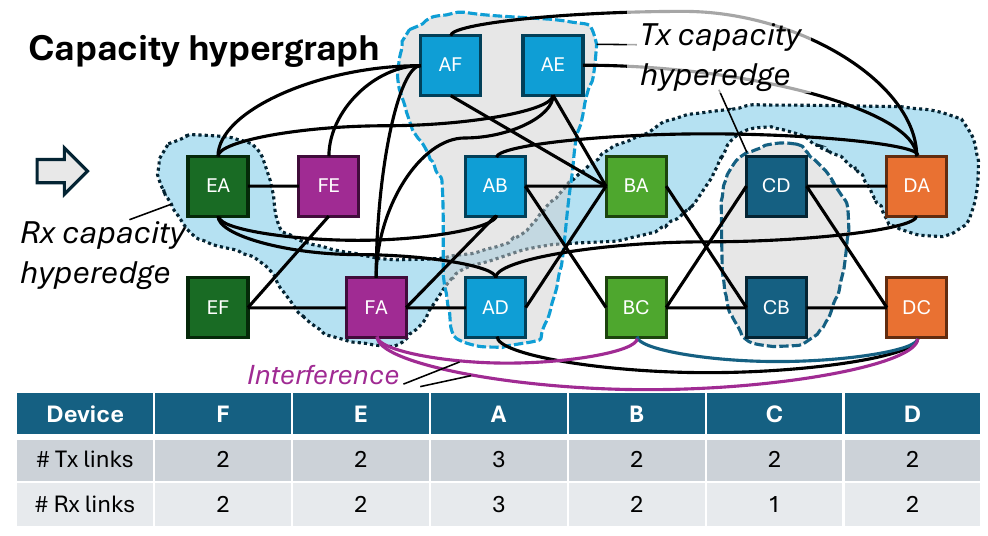}
    \label{fig:graph:global}
	  \vspace{-0.1in}
    }\hspace{-2mm}
    \subfloat[]{
    \includegraphics[height=1.65in]{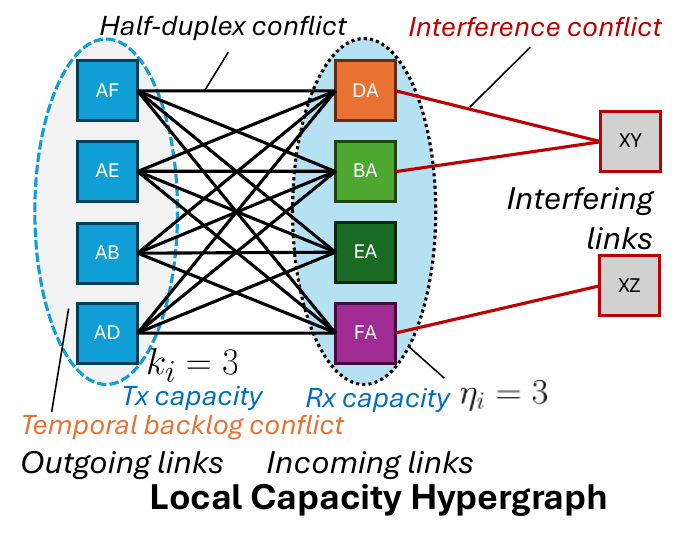}
    \label{fig:graph:local}
	  \vspace{-0.1in}
    }
	  \vspace{-0.15in}
	\caption{
		(a) A MIMO ad-hoc network, (b) its capacity hypergraph after internal contention and before OTA contention, along with the maximal numbers of outgoing and incoming links of each device listed. (c) A local capacity hypergraph on device A in distributed scheduling.
	} 
	\label{fig:graph}    
    \Description{(a) A MIMO ad-hoc network, (b) its capacity hypergraph after internal contention and before OTA contention, along with the maximal numbers of outgoing and incoming links, $\eta^{tx}_i$, $\eta^{rx}_i$, of each device listed. (c) A local capacity hypergraph on device A in distributed scheduling.}
\end{figure*}

In Figure.~\ref{fig:linksharing}, we present a minimal example of three devices to illustrate the difference in the initial assignment in the classic SP-BP operations in \eqref{E:commodity}-\eqref{E:quota} and our MaxU SP-BP in~\eqref{E:mu:select}--\eqref{E:sch-mimo}.
Notice that in both classic and MaxU SP-BP, the same packets maybe initially assigned to multiple links, which allows the same packets to choose different path.
However, this operation implicitly relies on the conflict graph for SISO networks as defined in Section~\ref{sec:sys}, i.e., only if all the outgoing links of a device are in conflict with each other, the constraint in~\eqref{E:formulation:comm} would be observed.
However, in MIMO networks, this could lead to the same packets being scheduled on multiple outgoing links of a device, causing detouring and looping. 
This problem is addressed in Section~\ref{sec:schedule} by  iteratively updating the preliminary rate assignment and link utilities during scheduling.

\subsection{Rule-based Initial Rate Assignment}
\label{sec:rule}
Besides MaxU, other rule-based approaches for multi-commodity selection can be developed to allocate the residual link rate after optimal commodity selection in classic SP-BP in Section~\ref{sec:sys:sp-bp}, e.g., based on sojourn time~\cite{ji2012delay,hai2018delay} and/or traffic priority.

\section{Distributed Scheduling for MIMO}\label{sec:schedule}
In a time-slotted MAC system, MIMO in networking can appear in multiple forms.
A multi-antenna transceiver can support concurrent outgoing or incoming links through spatial-division multiple access (SDMA), nullifying interference in certain directions with beamforming.
Multi-modal communications~\cite{kong2024covert} further introduce parallel links between a pair of transmitter and receiver.
Even a single antenna device can sequentially transmit small payloads to multiple receivers in a time slot, allowing them to decode their own data, which is referred as time-division multiple access (TDMA) -- although it may support only one incoming link per time slot, due to the challenge in coordinating multiple transmitters.

To extend MaxWeight scheduling to MIMO networks, we generalize the conflict graph in Section~\ref{sec:sys} into an attributed capacity hypergraph (ACH), and modify the distributed LGS scheduler~\cite{joo2012local} accordingly, as described in Algorithm~\ref{algo:lgs}. 

\subsection{Attributed Capacity Hypergraph}\label{sec:schedule:graph}
Consider each device $i\in\ccalV$ as a half-duplex transceiver with $\eta_i$ antennas, the schedule $\bbx(t)$ is subject to the following constraints
\begin{subequations}\label{E:io}
\begin{align}
 \bbx_{ia}(t) +  \bbx_{bi}(t)\leq 1,\;&\forall\; a,b\in \ccalN_{\ccalG^n}(i)\;.\label{E:io:half}    \\
 \sum_{j\in\ccalN_{\ccalG^n}(i)}\!\!\!\bbx_{ij}(t)\tau_{ij}(t)\leq \eta^{tx}_i  \;,&\sum_{j\in\ccalN_{\ccalG^n}(i)}\!\!\!\bbx_{ji}(t)\sigma_{ji}(t)\leq \eta^{rx}_i \;, \label{E:io:in}
\end{align}
\end{subequations}
In~\eqref{E:io}, $\tau_{ij}(t),\sigma_{ij}(t)$ are the costs of transmission and reception for link $(i,j)$, $\eta^{tx}_i,\eta^{rx}_i$ are the maximum number of transmission and reception streams.
The ACH, denoted by $\ccalG^c=(\ccalE, \ccalH)$, generalizes the definition of edge set $\ccalH$ in the conflict graph and introduces hyperedges with a capacity attribute.
A vertex $e\in\ccalE$ represents a directed link, and an attributed capacity hyperedge denoted as $h=(\eta^h: \{e_{1},\dots,e_{K}\})\in\ccalH$, which captures that the total cost of scheduled links in $h$ is limited to its capacity of $\eta^h$ per time slot.
For a regular pairwise edge, the capacity $\eta^h=1$.

\subsubsection{Interface Capacity}\label{sec:schedule:nodes}
We consider a transceiver $i\in\ccalV$ with $\eta_i\geq2$ antennas can support $\eta^{tx}_i=\eta^{rx}_i=\eta_i$ through SDMA, with $\tau_{ij}(t)=\sigma_{ij}(t)=1$.
A transceiver $i$ with $\eta_i=1$ can transmit on multiple outgoing links per time slot via TDMA, with $\tau_{ij}(t)\leq 1$ as air time, but receive on only one incoming link with $\sigma_{ij}(t)=1$.
For multi-modality links, we model each directed modality-link as a vertex in $\ccalG^c$, with interface capacity defined similarly as SDMA.
Therefore, for unit costs in~\eqref{E:io}, $\sigma_{ij}(t)=1$ for all $(i,j)\in\ccalE$, and   
\begin{equation}\label{E:cost}
    \tau_{ij}(t)\approx\begin{cases}
        \grave{\bbr}^{-1}_{ij}(t)\sum_{c\in\ccalC} \gamma_{ij}^{(c)}(t), & \eta_i=1\\
        1, &  \eta_i > 1
    \end{cases}\;,
\end{equation}

\subsubsection{Hyperedges}
There are three types of permanent conflict, 
a) half-duplex interface conflict in~\eqref{E:io:half};
b) interference conflict, e.g., two nearby links $e_1,e_2$ will interfere each other if activated simultaneously.
Both a) and b) can be represented by a regular pairwise edge, e.g., $(1:\{e_1,e_2\})$; and
c) transmission and reception capacity constraints on each device $i\in\ccalV$ in~\eqref{E:io:in}, respectively represented by attributed capacity hyperedges
$$ 
h^{tx}_i=\Big(\eta^{tx}_i: \{(i,j)\}_{j\in\ccalN_{\ccalG^n}(i)} \Big),\; 
h^{rx}_i=\Big(\eta^{rx}_i: \{(j,i)\}_{j\in\ccalN_{\ccalG^n}(i)} \Big)\;.
$$
In Figs.~\ref{fig:graph}, an exemplary MIMO ad-hoc network and its capacity hypergraph are illustrated, where the incoming links of single-antenna device $C$ conflict with each other since $\eta_C=1$.
Note that the interference can be asymmetric, e.g., $(F,A)$ interferes  incoming links of $C$, but $A$ can nullify the interference from $C$.

In addition to the permanent conflict relationships defined in~\eqref{E:io}, there is another temporary conflict relationship due to assigning the same packets to multiple outgoing links in both the exclusive SP-BP (\eqref{E:commodity}-\eqref{E:weight}) and MaxU SP-BP (\eqref{E:sort}-\eqref{E:allocate}).
Notice that such redundant initial assignment is intentional as it offers packets multiple options in moving forward.
However, this also means that we cannot directly activate multiple outgoing links with such redundant assignment, as it may violate constraint~\eqref{E:formulation:comm}.
To resolve this temporary backlog conflict, the scheduler must iteratively remove the scheduled backlogs from undecided links and recalculate the preliminary assignment $\gamma_{ij}^{(c)}(t)$ and link utility as the scheduling progresses.

If we leave the backlog conflict unresolved in MIMO networks, the same packets would likely be scheduled on multiple outgoing links of their host devices, allowing them to randomly choose suboptimal routes.
This not only wastes link bandwidth, but also causes looping, leading to significant increases in delay and a worse performance than that in SISO networks.

\begin{algorithm}[t]
\caption{Local Greedy Solver for ACH at Link-level (LGS-ACH)}
\label{algo:lgs-ach}
\begin{flushleft}
\hspace*{\algorithmicindent} 
\textbf{Input}: $\bbw,\grave{\bbr},\bbGamma, \bbQ, \bbU, \ccalG^n=(\ccalV,\ccalE),\ccalG^c=(\ccalE,\ccalH), K $ \\
\hspace*{\algorithmicindent} 
\textbf{Output}: $\bbM\in\reals^{|\ccalE|\times|\ccalC|}$, $\bbx\!\in\!\{0,1\}^{|\ccalE|}$ \COMMENTS{0: mute, 1: activate} 
\end{flushleft}
\begin{algorithmic}[1] 
\STATE $k=0, \bbx=-\boldsymbol{1}$, $\tilde{\bbQ}=\bbQ$, $\xi=\ccalE$, $\bbM = \big[\mu_{e}^{(c)} \mid e\in\ccalE, c\in\ccalC \big]$ 
\WHILE{ $\xi\neq\emptyset$ and $k\leq K$}
\STATE $k\leftarrow k +1$; update $\bbGamma,\bbw,\bbtau$ based on \eqref{E:gamma}, \eqref{E:cost}
\FORALL{ $e=(i,j)\in\ccalE$ and $\bbx_e = -1$ } 
\STATE \textbf{if} $ \tau_{e} > \eta^{tx}_i $ or $ \eta^{rx}_j=0 $ \textbf{then} $\bbx_e=0$, $\mu^{(c)}_{e}=0,\forall c\in\ccalC$;
\ENDFOR
\STATE $ \xi = \{e\mid e\in\ccalE, \bbx_e=-1\} $ \COMMENTS{set of undecided links}
\FORALL{ $e=(i,j)\in\xi $ } 
\IF{ $ e\in \textbf{top}_{\downarrow}(h^{tx}_i\cap\xi, w; 1) $ \COMMENTS{wins at Tx} }
\IF{ $ e\in \textbf{top}_{\downarrow}(\ccalN^1_{\ccalG^c}(e)\cap\xi, w; 1) $}
\STATE \COMMENTS{wins over half-duplex and interference edges}
\IF{ $ e\in \textbf{top}_{\downarrow}(h^{rx}_i\cap\xi, w; \eta^{rx}_i) $ }
\STATE \COMMENTS{wins over reception capacity hyperedge}
\STATE $\bbx_e=1$; $ \mu^{(c)}_{e}\!=\gamma^{(c)}_{e}\!,\; \tilde{Q}^{(c)}_{i}\!\!\leftarrow\! \tilde{Q}^{(c)}_{i}\!\!-\!\mu^{(c)}_{e}, \forall c\in\ccalC$ 
\STATE $\eta^{rx}_i\leftarrow\eta^{rx}_i-1,\; \eta^{tx}_i\leftarrow\eta^{tx}_i-\tau_{e}$ \COMMENTS{reduce cap.}
\ENDIF
\ENDIF
\ENDIF
\ENDFOR
\ENDWHILE
\end{algorithmic}
\end{algorithm}

\subsection{Local Greedy Scheduler for MIMO networks}
We present LGS-ACH in Algorithm~\ref{algo:lgs-ach}, which  generalizes the distributed MaxWeight link scheduler, LGS~\cite{joo2012local}, to our ACH conflict model at link level.
In each iteration, a link $(i,j)$ is scheduled only if it won contentions over all the incidental edges, including capacity hyperedge $h^{rx}_i$ and temporary backlog conflict at its transmitter, regular edges for half-duplex and interference conflicts, and reception capacity hyperedge $h^{rx}_i$.
Here, contention is resolved via an operation $ \epsilon=\textbf{top}_{\downarrow}(\varepsilon,w;n) $, defined as selecting the top-$n$ elements in set $\varepsilon=\{e\}$ ranked by weights $\{w_{e}\}$, and $\epsilon=\varepsilon$ if $|\varepsilon|\leq n$.
When a link is scheduled, the capacity of its incidental edges are updated, which leads to muting its conflicting neighbors.
Ties in contentions can be broken by comparing the unique IDs of equally weighted links.
This process terminates after all links have been decided.

\begin{algorithm}[t]
\caption{Local Greedy Scheduler for MIMO (LGS-MIMO)}
\label{algo:lgs}
\begin{flushleft}
\hspace*{\algorithmicindent} 
\textbf{Input}: $\bbw,\grave{\bbr},\bbGamma, \bbQ, \bbU, \ccalG^n=(\ccalV,\ccalE), \{\eta^{tx}_i,\eta^{rx}_i\}_{i\in\ccalV}, K $ \\
\hspace*{\algorithmicindent} 
\textbf{Output}: $\bbM\in\reals^{|\ccalE|\times|\ccalC|}$, $\bbx\!\in\!\{0,1\}^{|\ccalE|}$ \COMMENTS{0: mute, 1: activate}
\end{flushleft}
\begin{algorithmic}[1] 
\STATE $k=0, \bbx=-\boldsymbol{1}$, $\tilde{\bbQ}=\bbQ$,   $\bbv^{tx}=\bbv^{rx}=\boldsymbol{1}\in\reals^{|\ccalV|}$
\STATE $ \xi_i=\delta^+_{\ccalG^n}(i)\cup\delta^-_{\ccalG^n}(i), \phi_i=\emptyset,\forall i\in\ccalV$
\WHILE{ $\bbv^{tx}+\bbv^{rx}\neq \boldsymbol{0}$ and $k\leq K$}
\STATE $ k\leftarrow k+1 $
\FORALL{ $i\in\ccalV$ and $\bbv^{tx}_i+\bbv^{rx}_i>0$ } 
\STATE \COMMENTS{Intra-transmitter contention}
\IF{ $ \delta^-_{\ccalG^n}(i)\cap \xi_i \neq \emptyset $ }
\FOR{ $ (i,j)\in\delta^-_{\ccalG^n}(i)\cap \xi_i $ }
\STATE Update $ \gamma_{ij}^{(c)}$,  $w_{ij}$, $\tau_{ij}$ based on \eqref{E:gamma}, \eqref{E:cost}
\STATE \textbf{if} $\tau_{ij} > \eta^{tx}_i$ \textbf{then} $ \bbx_{ij}=0 $, $ \xi_i \leftarrow \xi_i\setminus\{(i,j)\} $
\ENDFOR
\STATE $e_i=\textbf{top}_{\downarrow}\big(\delta^-_{\ccalN^n}(i)\cap\xi_i, w; 1\big) $ \COMMENTS{Select outgoing link}
\STATE Transmit a RTS $\{e_i:w_{e_i}\}$ to its intended receiver
\ELSE 
\STATE $\bbv^{tx}_i=0$
\ENDIF
\STATE \COMMENTS{Exchange of RTS messages}
\STATE Receive RTS from all nearby devices, denoted by $\hat{\ccalN}_{\ccalH}(i)$ 
\STATE $ \ccalG^c_i, \bbomega^i = \varrho\big( \delta^+_{\ccalG^n}(i),\{e_j:w_{e_j}\}_{j\in\hat{\ccalN}_{\ccalH}(i)\cup\{i\}}\big) $ with Algo.~\ref{algo:lch}
\STATE $\theta_i=\emptyset$; {find greedy MWIS $\varepsilon_i=\varphi(\ccalG^c_i,\bbomega^i)$} with Algo.~\ref{algo:greedy}
\IF{$\eta^{rx}_i>0$ and $ e_i\notin \varepsilon_i $}
\STATE $\theta_i \!=\! \textbf{top}_{\downarrow}(\varepsilon_i, w; \eta^{rx}_i) $,\; $\phi_i \!\leftarrow\!\phi_i\cup\ccalN^1_{\ccalG^c_i}(\theta_i) $, 
\STATE $\eta^{rx}_i\!\leftarrow\!\eta^{rx}_i\!-\!|\theta_i|$,\; $\xi_i\leftarrow\xi_i \setminus \theta_i $
\STATE \textbf{if} $\bbv^{tx}_i=1$ \textbf{then} $\bbx_{e_i}=0$ \COMMENTS{mute $e_i$}
\ENDIF
\IF{ $\eta^{rx}_i=0$ or $\delta^+_{\ccalG^n}(i) \cap\xi_i=\emptyset $ }
\STATE $\bbv^{rx}_i=0$, $\phi_i\leftarrow\phi_i\cup (\delta^+_{\ccalG^n}(i) \cap\xi_i)$ 
\ENDIF
\STATE Broadcast a CTS $(\theta_i,\phi_i)$ to nearby devices $j\in\hat{\ccalN}_{\ccalH}(i)$
\STATE \COMMENTS{Exchange of CTS messages}
\IF{ $\bbv^{tx}_i=1$}
\STATE Receive CTS $(\theta_j,\phi_j), \forall j\in\hat{\ccalN}_{\ccalH}(i)$ from nearby devices
\STATE $ \Theta_i= \bigcup_{j\in\hat{\ccalN}_{\ccalH}(i)} \theta_j$,\; $ \Phi_i= \bigcup_{j\in\hat{\ccalN}_{\ccalH}(i)} \phi_j$
\IF{ $ e_i\in \Phi_i $ \COMMENTS{Rejected by some Rx}}
\STATE $\bbx_{e_i}=0 $;\; $ \mu^{(c)}_{e_i}=0,\forall c\in\ccalC $;\; $ \xi_i \leftarrow \xi_i\setminus\{e_i\} $
\ELSIF{ $ e_i\in \Theta_i $ \COMMENTS{Granted by intended Rx}}
\STATE $\bbx_{e_i}=1 $;\; $\mu^{(c)}_{e_i}=\gamma^{(c)}_{e_i},\;\tilde{Q}^{(c)}_{i}\leftarrow \tilde{Q}^{(c)}_{i}-\mu^{(c)}_{e_i},\forall c\in\ccalC$
\STATE $\eta^{tx}_i\leftarrow\eta^{tx}_i-\tau_{e_i}$,\; $ \phi_i\leftarrow\phi_i\cup \delta^+_{\ccalG^n}(i) $
\ENDIF
\STATE $\phi_i\leftarrow \phi_i \cup\{(j',i)\mid (j',i)\in \Phi_i\}$  \COMMENTS{Avoid hidden node problem by including rejected incoming links to $\phi_i$}
\ENDIF
\ENDFOR
\ENDWHILE
\end{algorithmic}
\end{algorithm}

To remove the temporary backlog conflicts between scheduled and unscheduled links, we update the residual backlog state $\tilde{\bbQ}$ in line 14 in Algo.~\ref{algo:lgs-ach}, and update the utilities and transmission costs of undecided links accordingly in each iteration (line 3 in Algo.~\ref{algo:lgs-ach}):
\begin{equation}\label{E:gamma}
    \gamma^{(c)}_{ij}\leftarrow\min\left\{\tilde{Q}^{(c)}_{ij}, \gamma^{(c)}_{ij}\right\}\;,\;
    w_{ij} =\sum_{c\in\ccalC} \gamma_{ij}^{(c)}\cdot U_{ij}^{(c)}(t).
\end{equation}
Here, in order to differentiate the commodity selection operations and for clarity, we drop $(t)$ for operations inside the scheduler.
Following~\eqref{E:gamma}, the unit cost $\tau_{ij}$ is updated with~\eqref{E:cost}.

Our link-level LGS-ACH can be translated into a protocol-style transceiver-level operations of distributed link scheduling for MIMO networks, denoted as LGS-MIMO and detailed in Algo.~\ref{algo:lgs}.
In an iteration, each transceiver $i\in\ccalV$ first selects an outgoing link $e_i$ based on cost and utility, then transmits a request-to-send (RTS) message to its intended receiver enclosing the identity and utility of $e_i$.
Next, device $i\in\ccalV$ receives RTS messages from its nearby devices, denoted by set $ \hat{\ccalN}_{\ccalH}(i) $.
With Algorithm~\ref{algo:lch} (Appendix \ref{app:algos}), these RTS messages are used to build a local conflict graph, $\ccalG^c_i=(\ccalE_i,\ccalH_i)$, which is a subgraph of the local capacity hypergraph  exemplified in Fig.~\ref{fig:graph:local}.
Then, with greedy MWIS heuristic in Algorithm~\ref{algo:greedy} ( Appendix \ref{app:algos}), each transceiver $i$ finds a greedy MWIS $\varepsilon_i$ on the $\ccalG^c_i$  (line 20 in Algo.~\ref{algo:lgs}), and selects a set of incoming links within its residual reception capacity as $\theta_i \!=\! \textbf{top}_{\downarrow}(\varepsilon_i, w; \eta^{rx}_i) $, and adds their local conflicting neighbors, denoted by $\ccalN_{\ccalG^c_i}(\theta_i)$, to its rejection list $\phi_i$.
By exchanging a clear-to-send (CTS) message enclosing $(\theta_i,\phi_i)$ with nearby devices, device $i$ learns if its transmission request on outgoing link $e_i$ is rejected or accepted.
If $e_i$ is accepted by its intended receiver, device $i$ marks $e_i$ as scheduled and adds its incoming links to $ \phi_i $ to inform their transmitters in the next round.
If $e_i$ is rejected by the intended or any unintended receivers, it is muted by device $i$.
Otherwise, $e_i$ enters the next iteration.
To avoid the hidden node problem, transceiver $i$ also add its incoming links that are rejected by other interfered devices to $\phi_i$ to inform their transmitters.
To limit the size of $\phi_i$, a rejected link can be removed from $\phi_i$ once $i$ receives an acknowledgment from its transmitter.
If link $e_i$ is scheduled, its transmitter $i$ updates its residual backlog state $\tilde{Q}^{(c)}_{i}$ for all $c\in\ccalC$, and the link utility and cost of any undecided outgoing links (line 9). 

LGS-MIMO is fully distributed without requiring the global knowledge of $\ccalG^c$ on any devices.
By limiting the maximum size of active ACH in LGS-ACH from $|\ccalE|$ to $|\ccalV|$ in each iteration, it reduces local exchange complexity from $\ccalO(\log|\ccalE|)$~\cite{ji2012delay} to $\ccalO(\log|\ccalV|)$.

\section{Dominance over Classic SP-BP}\label{sec:proof}
We can prove that MaxU SP-BP dominates classic SP-BP.
\begin{theorem}\label{th:dominance}
With everything else being equal, MaxU SP-BP is dominant over the classic SP-BP in the network capacity region.
\end{theorem}
\begin{proof}
We mark the link-commodity rate assignments of the classic SP-BP as $\hat{\bbM}(t)$ and that of MaxU SP-BP as $\tilde{\bbM}(t)$, and likewise for other intermediate variables.
Under heavily loaded traffic, e.g., $\grave{\bbr}_{ij}(t) \leq Q_{i}^{(c_{ij}^{*}(t))}\!(t)$,  steps 1-2 in~\eqref{E:mu:select}--\eqref{E:sch-mimo} are equivalent to those in~\eqref{E:commodity}-\eqref{E:weight}, i.e., $\hat{\gamma}^{(c)}_{ij}(t)=\tilde{\gamma}^{(c)}_{ij}(t),\forall\; (i,j)\in\ccalE,c\in\ccalC$, as illustrated in Fig.~\ref{fig:linksharing}.
In this case, MaxU and classic SP-BP are equivalent. 

Under lightly loaded traffic, $\grave{\bbr}_{ij}(t) > Q_{i}^{(c_{ij}^{*}(t))}\!(t)$, the preliminary rate assignment under MaxU SP-BP always includes that of the optimal commodity under the classic SP-BP, given the same queueing state $\bbQ(t)$ and biases $\bbB$.
Without backlog conflicts, e.g., in SISO networks or MIMO networks with LGS-MIMO, we have 
$$
\tilde{\gamma}^{(c)}_{ij}(t)\begin{cases}
    = \hat{\gamma}^{(c)}_{ij}(t), & c=c^*_{ij}(t)\\ 
    \geq\hat{\gamma}^{(c)}_{ij}(t)=0, & c\neq c^*_{ij}(t), U^{(c)}_{ij}(t) >0\\ 
    = \hat{\gamma}^{(c)}_{ij}(t) =0, & c\neq c^*_{ij}(t), U^{(c)}_{ij}(t) \leq 0
\end{cases}\;,\;\forall (i,j)\in\ccalE,
$$ and based on~\eqref{E:weight}, $\tilde{w}_{ij}(t)\geq\hat{w}_{ij}(t)$, $\forall\; (i,j)\in\ccalE$.
Therefore, $$ \mathbb{E}\Bigg[\sum_{(i,j)\in\ccalE} \tilde{\bbx}_{ij}(t)\tilde{w}_{ij}(t)\Bigg]\geq\mathbb{E}\Bigg[\sum_{(i,j)\in\ccalE} \hat{\bbx}_{ij}(t)\hat{w}_{ij}(t)\Bigg]\;,$$
under the same $\ccalG^c$, since $w_{ij}(t)$ is generally independent from the topology of $\ccalG^c$.
With $\mu^{(c)}_{ij}(t)=\bbx_{ij}(t)\gamma^{(c)}_{ij}(t)$ and \eqref{E:weight}, the expectation of the objective function in~\eqref{E:formulation} under MaxU SP-BP is greater than or equal to that under the classic SP-BP:
\begin{equation*} 
    \mathbb{E}\Bigg[\sum_{(i,j)\in\ccalE}\sum_{c\in\ccalC} \tilde{\mu}_{ij}^{(c)}\!(t) U_{ij}^{(c)}\!(t)\Bigg] \!\geq\! \mathbb{E}\Bigg[\sum_{(i,j)\in\ccalE}\sum_{c\in\ccalC} \hat{\mu}_{ij}^{(c)}\!(t) U_{ij}^{(c)}\!(t)\Bigg].
\end{equation*}
For classic SP-BP that directly applies steps 3-4 in \eqref{E:scheduling}-\eqref{E:quota} to MIMO networks without resolving the backlog conflicts, where $\check{\mu}^{(c)}_{ij}(t)$ stands for its actual rate assignment in the queueing state transition in~\eqref{E:queue} after one of the outgoing links scheduled for a packet is selected randomly, we have $ \hat{\mu}_{ij}^{(c_{ij}^{*}(t))}\!(t) \geq \check{\mu}^{(c_{ij}^{*}(t))}_{ij}(t) $, therefore 
$$ 
\mathbb{E}\Bigg[\sum_{(i,j)\in\ccalE}\sum_{c\in\ccalC} \hat{\mu}_{ij}^{(c)}\!(t) U_{ij}^{(c)}\!(t)\Bigg] \geq 
\mathbb{E}\Bigg[\sum_{(i,j)\in\ccalE}\sum_{c\in\ccalC} \check{\mu}_{ij}^{(c)}\!(t) U_{ij}^{(c)}\!(t)\Bigg]\;. 
$$

Based on the proof of maximum queueing stability (throughput optimality) for classic BP (and SP-BP) in~\cite{zhao2024tmlcn,neely2005dynamic}, a larger objective value in~\eqref{E:formulation} corresponds to a tighter upper bound of the Lyapunov drift of the backlog, in other words, a larger region of exogenous packet arrival rate with queueing stability.
Therefore, MaxU SP-BP dominates classic SP-BP in the network capacity region. 
\end{proof}

\section{Numerical experiments}
\label{sec:results}
We evaluate four SP-BP variations based on the combinations of two commodity selection approaches (\textit{Excl} for the classic exclusive selection and \textit{MaxU}) and two optimal shortest path (SP) bias schemes in~\cite{zhao2024tmlcn}:
1) SP-$\bar{r}$, in which edge weight in the connectivity graph $\ccalG^n$ is set to be the network-wide average long-term link rate $\bar{r}$; and
2) SP-$\bar{r}r_{max}/r$, in which the weight of a link $e$ in $\ccalG^n$ is set as $\bar{r}r_{max}/r_e$, where $r_{max}$ is the maximum long-term link rate across the network\footnote{Source code: \url{https://github.com/zhongyuanzhao/MaxU-SPBP}}.
The bias matrix $\bbB$ is found by all-pairs shortest path algorithm on edge-weighted connectivity graph~\cite{bernstein2019distributed}.

The simulation is based on the same set of network instances in \cite{zhao2024tmlcn}, generated by a 2D point process with different network sizes $|\ccalV|\in\{20,\dots,100\}$.
For each $|\ccalV|$, 100 test instances are generated by drawing 10 instances of random networks, each with 10 realizations of random flows and random link rates. 
A test instance contains $ 0.4|\ccalV|$ flows with distinct source-destination pairs,  uniformly distributed long-term link rates, $r_{e} \sim \mathbb{U}(10,42)$, and real-time link rates as $\grave{\bbr}_{e}(t) \sim  \mathbb{N}(r_{e}, 3)$, truncated to between $r_e\pm 9$, to capture fading channels with lognormal shadowing.
The simulation horizon is set as $T=1000$ time steps.
Unless otherwise stated, a flow can be configured at equal chances with either 1)~\textit{streaming} traffic, where new packet arrivals at the source follow a Poisson process with a uniformly random arrival rate $\lambda_c \sim \mathbb{U}(0.1,1.0)$ throughout the simulation horizon; or
2) \textit{bursty} traffic, new packet arrivals follow the same Poisson process but only lasting for $30$ time slots with a random start time $\mathbbm{U}(0,T-100)$.  
In this setting, the total number of packets of a bursty flow is far less than that of a streaming flow, presenting the most challenging scenarios for evaluating the ability of SP-BP in mitigating the last-packet problem.

Two types of networks are considered, 1) in SISO configuration, all mobile devices are equipped with omnidirectional antennas and transmit at identical power levels, 2) in MIMO configuration, the number of antennas of each device follows a distribution of $P(1)=0.2,P(2)=0.5,P(3)=0.2,P(4)=0.1$.

\begin{figure}[!t]
	\centering
	\includegraphics[width=\linewidth]{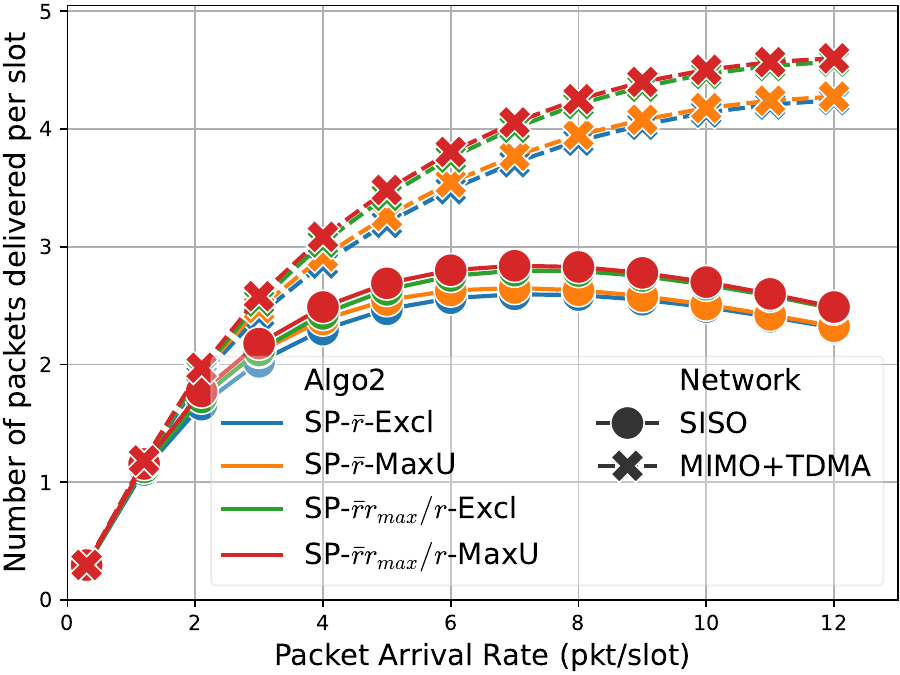}
	\caption{Average throughput per flow versus flow rate $\lambda$ in networks of 100 nodes. All flows are streaming at identical arrival rate.
 }
	\label{fig:thpt}
    \Description{Average flow throughput versus flow rates in networks of 100 nodes. All flows are streaming at identical arrival rate.}
\end{figure}

\subsection{Throughput Test}\label{sec:results:thpt}
We first test the throughput performance of the four SP-BP schemes in SISO and MIMO networks of a fixed size $|\ccalV|=100$ by setting all flows as streaming traffic with identical flow rate $\lambda$. 
In Fig.~\ref{fig:thpt}, the flow throughput (the average number of delivered packets per time slot) versus the flow rate $\lambda$ under these schemes are presented.
Our scheduling approach, LGS-MIMO, is demonstrated to be effective in MIMO networks as evident from the significantly increased throughput of all SP-BP schemes compared to that of SISO networks.
The SP-$\bar{r}r_{max}/r$ also increases the throughput over SP-$\bar{r}$ with more link features, consistent with the result in~\cite{zhao2024tmlcn}.
Most importantly, with everything else being equal, our MaxU always achieves slightly higher throughput than the exclusive commodity selection in the classic SP-BP.
This empirical result supports our theoretical proof in Section~\ref{sec:proof} that MaxU SP-BP dominates the classic SP-BP in the network capacity regime.
Notice that our classic SP-BP baselines with optimized biases already significantly outperform the original and many improved BP schemes in throughput as shown in~\cite{zhao2024tmlcn}.

\begin{figure}[!t]
	\centering
	\hspace{-3mm}
	\subfloat[]{
		\includegraphics[height=2.5in]{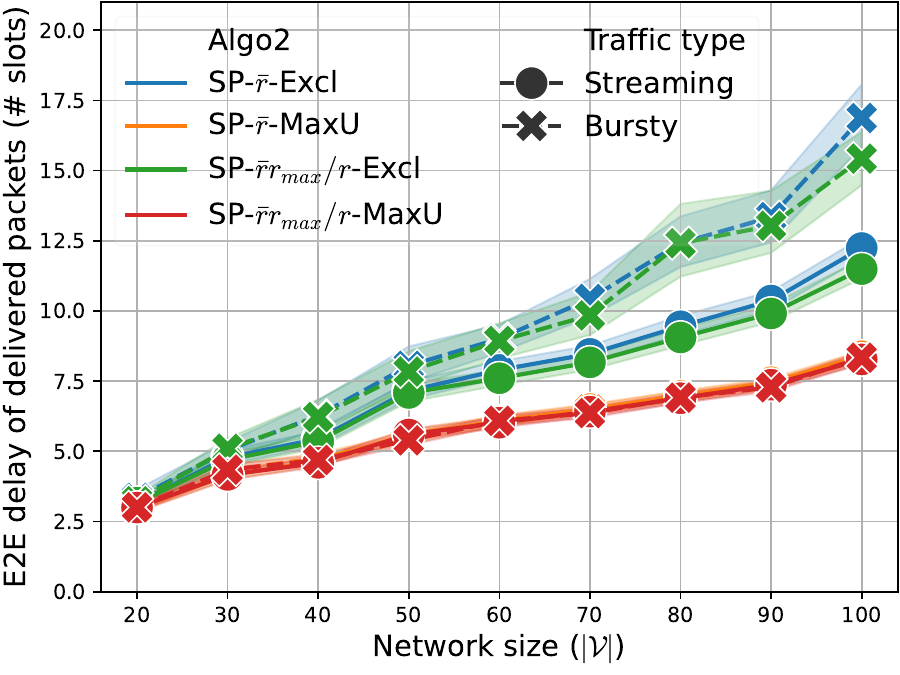}
		\label{fig:mimo:delay}\vspace{-0.1in}
	}\\
    \hspace{-3mm}
	\subfloat[]{
		\includegraphics[height=2.5in]{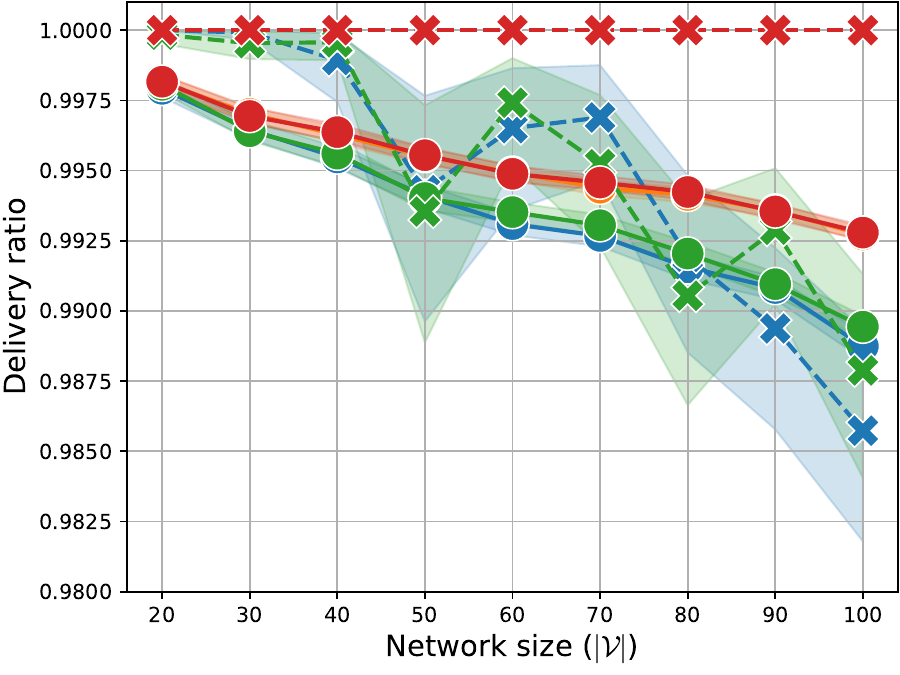}
		\label{fig:mimo:delivery}\vspace{-0.1in}
	}
	\caption{
        Results on networks of 20-110 nodes with MIMO and TDMA for $T=1000$:
        (a)~Average end-to-end delay of delivered packets,
        (b)~Average packet delivery ratio at the end of simulation.
        The bands indicate $95\%$ confidence interval. 
        There are $0.4|\ccalV|$ commodities, flow rate $\lambda_c\in\mathbbm{U}(0.1,1)$.
        At a probability of $0.5$, a flow can be configured with streaming traffic or bursty traffic lasting $30$ time slots with a random start time $\mathbbm{U}(0,T-100)$.
	}
	\label{fig:mimo}    
    \Description{test results curves}
\end{figure}

\subsection{Bursty vs Streaming in MIMO Networks}

Next, the four SP-BP schemes are evaluated in networks of different sizes under the mixed traffic setting as described in the beginning of this section.
The traffic is lightweight with flow rate $\lambda_c\in\mathbb{U}(0.1,1)$ compared to the maximum throughput $ > 2.5$ for SISO and $>4.2$ for MIMO, as shown in Fig.~\ref{fig:thpt}.
We track individual packets and collect three performance metrics for each flow: 
the average end-to-end latency (no. of slots) and trip length (no. of hops) of packets that have reached their destinations, and the delivery ratio between the numbers of delivered and injected packets.
Next, we gather the performance metrics of the average flows (mean metric across all flows) and $95$ percentile flows ($95$ percentile of all flows) for each test instance, and plot the mean metric of 100 instances of the same network size $|\ccalV|$ as a single point in Figs~\ref{fig:mimo},~\ref{fig:network}, and~\ref{fig:hops}.

In Figs.~\ref{fig:mimo}, we compare the latencies and delivery ratios of streaming and bursty traffic of average flows under four SP-BP schemes in MIMO networks.
With exclusive commodity selection, bursty traffic suffers up to $33-37\%$ higher latency on average (Fig.~\ref{fig:mimo:delay}).
In comparison, MaxU SP-BP achieves almost identical latency for streaming and bursty traffic,
which is also $32-37\%$ lower than those of the streaming traffic under classic SP-BP. 
By examining the average delivery ratios in Figs.~\ref{fig:mimo:delivery}, we find that MaxU delivers $100\%$ packets for bursty flows, and increases delivery ratio of streaming flows.
The delivery ratio of bursty traffic under classic SP-BP decreases as network grows larger.
It demonstrates that link-sharing MaxU is more effective in bandwidth utilization and addressing last-packet problem than exclusive commodity selection.

\begin{figure}[!t]
	\centering
    \hspace{-3mm}
	\subfloat[Average flows]{
		\includegraphics[height=2.5in]{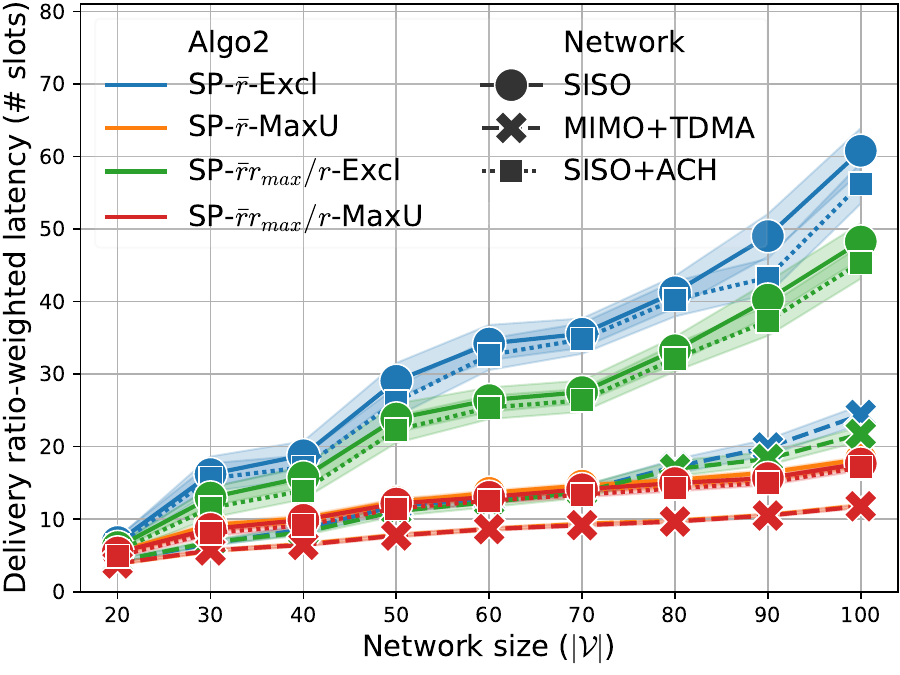}
		\label{fig:network:e2e}\vspace{-0.1in}
	}\\
    \hspace{-3mm}
	\subfloat[$95$ percentile flows]{
		\includegraphics[height=2.5in]{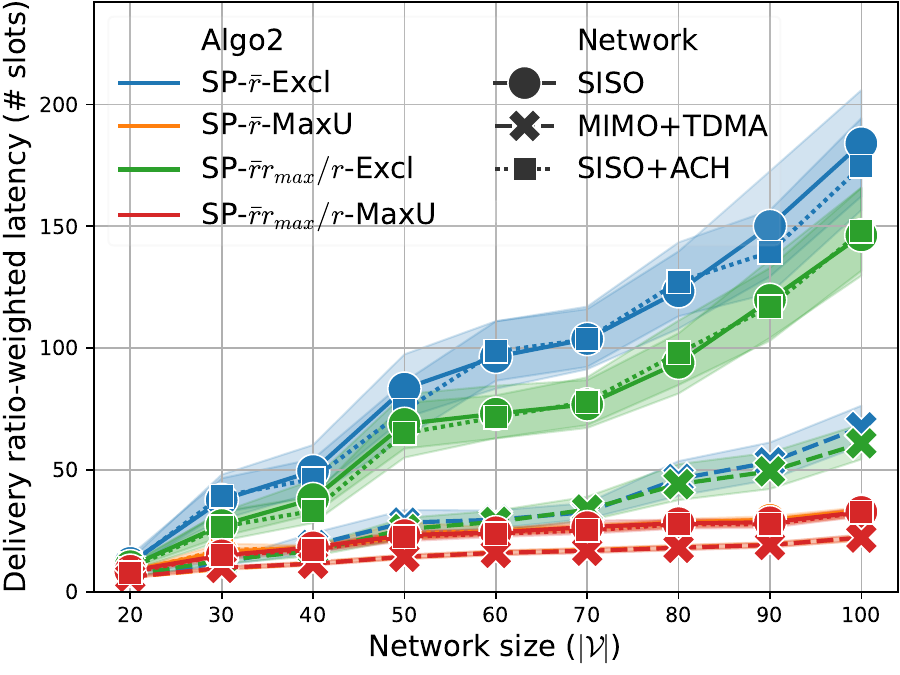}
		\label{fig:network:e2e95}\vspace{-0.1in}
	}
	\caption{
        Latency weighted by delivery ratio (Latency $\times$ delivery ratio  + $T$($1-$ delivery ratio)) in networks of 20-110 nodes for $T=1000$:
        (a)~Average flows, and
        (b)~$95$ percentile flows (y axis range tripled). 
        The bands indicate $95\%$ confidence interval. 
        Same test configuration as that in Figs.~\ref{fig:mimo}.
	}
	\label{fig:network}    
    \Description{test results curves}
\end{figure}

\subsection{Link Scheduling in SISO vs MIMO}
For behavioral insights of the commodity selection approaches and distributed schedulers in SISO and MIMO networks, we present the latency weighted by delivery ratio, and packet trip length for the average and $95$ percentile flows regardless of traffic types in Figs.~\ref{fig:network}.
Latency weighted by delivery ratio is introduced in~\cite{zhao2024tmlcn} to aggregate the latency measurement and delivery ratio by treating the latency of undelivered packets as $T=1000$, which enables a consistent comparison against other BP schemes evaluated in~\cite{zhao2024tmlcn}. 

For average flows (Fig.~\ref{fig:network:e2e}), MaxU can reduce the composite latency of exclusive commodity selection by up to $70\%$ in SISO networks, and $50\%$ in MIMO networks.
The benefit of MaxU is more significant in $95$ percentile flows, representing the worst cases, reducing the composite latency in SISO networks by up to $80\sim85\%$, and by $60\%$ in MIMO networks, as illustrated in Fig.~\ref{fig:network:e2e95}.

In addition, we plug our LGS-MIMO into SISO networks by setting antenna number as 1 for all devices, annotated as SISO+ACH, which performs similarly with the original LGS~\cite{joo2012local} with even slightly better latency and packet trip length. 
This shows that the original LGS operating on conflict graphs can be viewed as a special case of LGS-MIMO in SISO networks.

\begin{figure}[!t]
	\centering
    \hspace{-3mm}
	\subfloat[Average flows]{
		\includegraphics[height=2.5in]{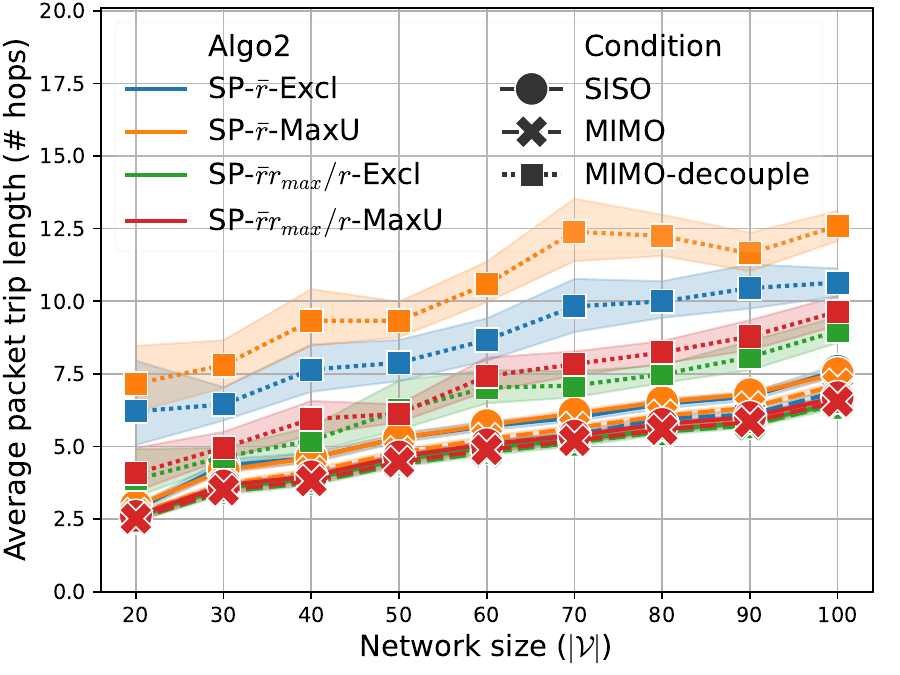}
		\label{fig:hops:mean}\vspace{-0.1in}
	}\\
    \hspace{-3mm}
	\subfloat[$95$ percentile flows]{
		\includegraphics[height=2.37in]{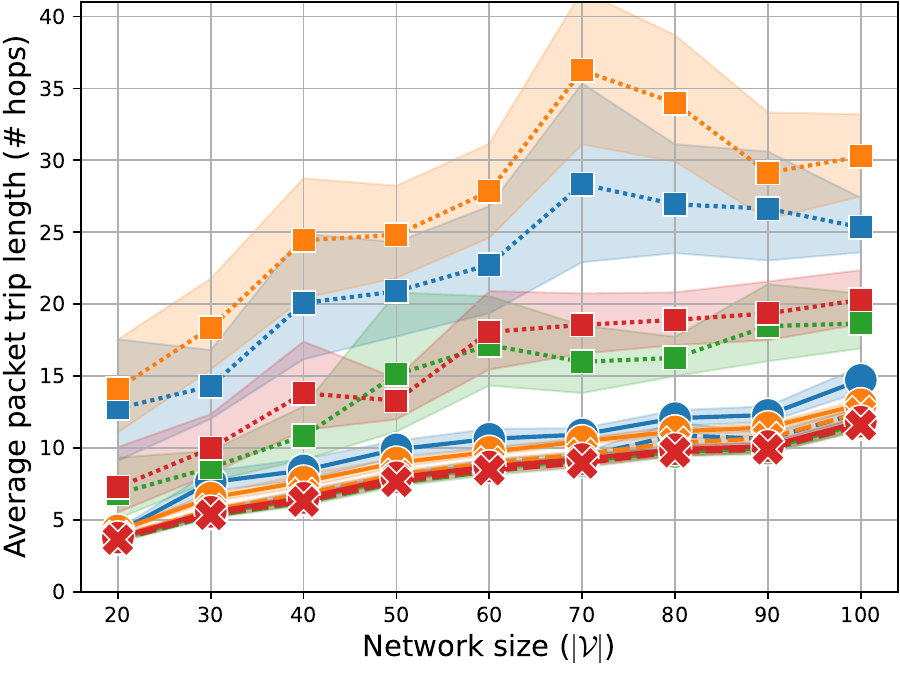}
		\label{fig:hops:q95}\vspace{-0.1in}
	}
	\caption{
        Average trip length of delivered packets in networks of 20-110 nodes for $T=1000$:
        (a)~Average flows.
        (b)~$95$ percentile flows (y axis range doubled).
        The values under MIMO-decouple decrease in $|\ccalV|\in\{80,90,100\}$ due to lower delivery ratio.
        The bands indicate $95\%$ confidence interval. 
        Same test configuration  as that in Figs.~\ref{fig:mimo}.
	}
	\label{fig:hops}    
    \Description{test results curves}
\end{figure}

To reveal the factors that influence the latency, we present the mean packet trip length for the average and $95$ percentile flows in Figs.~\ref{fig:hops:mean} and \ref{fig:hops:q95}, respectively. 
Notice the difference between trip length and path length, as the former accounts for the actual trips of individual packets including detours and loops.
The result shows that better bias SP-$\bar{r}r_{max}/r$ and/or MIMO networks can reduce packet traversal distance compared to the baselines of SP-$\bar{r}$ in SISO networks. 
Whereas the MaxU commodity selection results in similar or even slightly longer trips compared to classic exclusive approach. 
It indicates that the benefit of MaxU comes from efficient bandwidth utilization  rather than better routes.

We also include an ablation test in MIMO networks by decoupling the initial rate assignment and link scheduling, as illustrated by dotted lines with square markers (MIMO-decouple) in Figs.~\ref{fig:hops:mean} and \ref{fig:hops:q95}. 
It can be observed that without rate reassignment during scheduling, the trip length of both average and 95 percentile flows increase significantly in MIMO networks, with the  latter suffer most.
It shows that without LGS-MIMO, all SP-BP schemes suffer from detouring and looping, with link-sharing commodity selection making it worse, demonstrating the critical role of rate reassignment for scheduling in MIMO networks.

\section{Conclusions}
\label{sec:conclusions}

In this work, we address the new challenges of backpressure routing  since its invention brought by diverse traffic and advances in air-interfaces.
In particular, short-lived traffic, higher link bandwidth, and MIMO capabilities based on TDMA, SDMA, and multi-modality communications challenge the core operations in BP schemes.
Our solutions in commodity link sharing and distributed link scheduling renew the core operations of classic backpressure routing for modern wireless networks, based on the same theoretical foundation.
The benefits of our solutions in latency and throughput over classic approaches are supported by both theoretical analysis and experimental results. 
More importantly, our methods can be incorporated into many existing modifications of BP schemes for various scenarios, and could help popularize SP-BP routing in practice.
Future works include further understanding its impact on routing performance in various network connectivity and interference settings, enriching conflict modeling for multi-modal and full-duplex transceivers, and support for traffic of different priorities, and different queueing disciplines.

\begin{acks}
Research was sponsored by the DEVCOM ARL Army Research Office and was accomplished under Cooperative Agreement Number W911NF-24-2-0008. 
The views and conclusions contained in this document are those of the authors and should not be interpreted as representing the official policies, either expressed or implied, of the DEVCOM ARL Army Research Office or the U.S. Government. 
The U.S. Government is authorized to reproduce and distribute reprints for Government purposes notwithstanding any copyright notation herein.
\end{acks}

{
\bibliographystyle{ACM-Reference-Format}
\bibliography{strings,refs,refs_ol}
}

\appendix 
\section{Algorithms for local conflicts}\label{app:algos}

\begin{algorithm}[ht]
\caption{Construct Local Conflict Graph $ \varrho(\cdot) $}
\label{algo:lch}
\begin{flushleft}
\hspace*{\algorithmicindent} 
\textbf{Input}: $\delta^+_{\ccalG^n}(i), \{e_j:w_{e_j}\}_{j\in \hat{\ccalN}_{\ccalH}(i)\cup\{i\}} $ \\
\hspace*{\algorithmicindent} 
\textbf{Output}: $\ccalG^c_i=(\ccalE_i,\ccalH_i), \bbomega^i$     
\end{flushleft}
\begin{algorithmic}[1] 
\STATE $\ccalE_i= \{e_i\}$, $\omega^i_i=w_{e_i}$, $\ccalH_i=\emptyset $ 
\FOR{ $ j\in \hat{\ccalN}_{\ccalH}(i)  $}
\IF{ $ e_j \in \delta^+_{\ccalG^n}(i) $ or $j$ is a known interferer}
\STATE $ \ccalE_i\leftarrow\ccalE_i\cup\{e_j\} $, $\omega^i_{e_j}=w_{e_j}$
\STATE $ \ccalH_i\leftarrow\ccalH_i\cup( 1:\{e_i, e_j\} ) $ \COMMENTS{half-duplex or interference}
\ELSE 
\STATE Analyze signal properties (e.g., DoA, RSS) of RTS for $e_j$
\FOR{ $ e_{j'}\in\delta^+_{\ccalG^n}(i) $ and $e_{j'}$ would be interfered by $e_j$}
\STATE \textbf{if} $ e_j\notin\ccalE_i $ \textbf{then} $ \ccalE_i\leftarrow\ccalE_i\cup\{e_j\} $, $\omega^i_{e_j}=w_{e_j}$
\STATE $ \ccalH_i\leftarrow\ccalH_i\cup( 1:\{e_{j'}, e_j\} ) $ \COMMENTS{new interference}
\ENDFOR
\ENDIF
\ENDFOR
\end{algorithmic}
\end{algorithm}

\begin{algorithm}[ht]
\caption{Greedy MWIS on Local Conflict Graph $\varphi(\cdot) $}
\label{algo:greedy}
\begin{flushleft}
\hspace*{\algorithmicindent} 
\textbf{Input}: $\ccalG^c_i=(\ccalE_i,\ccalH_i), \bbomega^i $ \\
\hspace*{\algorithmicindent} 
\textbf{Output}: independent set $\varepsilon_i$     
\end{flushleft}
\begin{algorithmic}[1] 
\STATE $\varepsilon_i =\emptyset$; $m_e=-1,\forall e \in \ccalE_i$; $\tilde{\bbe}=\arg\textup{sort}_{\downarrow}\left[\omega^i_{e}\right]_{e\in\ccalE_i}  $
\FOR{ $n=1,\dots,|\ccalE_i|$ } 
\IF{ $m_{ \tilde{\bbe}_n}=-1$ and $\omega^i_{\tilde{\bbe}_n}>0$ }  
\STATE $\varepsilon_i \leftarrow\varepsilon_i \cup \{\tilde{\bbe}_n\} $ 
\STATE $ m_{e}=0 , \forall\; e\in \ccalN^1_{\ccalG^c_i}(\tilde{\bbe}_n)$ \COMMENTS{mute all regular neighbors}
\ENDIF
\ENDFOR
\end{algorithmic}
\end{algorithm}

\end{document}